\documentclass[10pt,twocolumn]{article}
\usepackage{mya4}
\settextwidth{175mm}
\settextheight{240mm}
\usepackage{graphics}
\usepackage{amssymb}
\usepackage{amsmath}
\usepackage{stmaryrd}

\newtheorem{lemma}{Lemma}

\newtheorem{theorem}{Theorem}
\newtheorem{definition}{Definition}
\newtheorem{notation}{Notation}
\newenvironment{proof}{\textit{Proof:}}{}
\newcommand{\qed}{\hfill$\Box$}

\newcommand{\ket}[1]{|#1\rangle}

\newcommand{\qgate}[1]{\mathsf{#1}}

\newcommand{\matr}[4]{\begin{pmatrix}{#1}&{#2}\\{#3}&{#4}\end{pmatrix}}

\newcommand{\bra}[1]{\langle#1|}

\newcommand{\gX}{\qgate{X}}
\newcommand{\gY}{\qgate{Y}}
\newcommand{\gZ}{\qgate{Z}}
\newcommand{\gH}{\qgate{H}}
\newcommand{\gCnot}{\qgate{CNot}}
\newcommand{\gP}{\qgate{P}}

\newcommand{\ghz}{\mathsf{GHZ}}
\newcommand{\ighz}{\mathsf{iGHZ}}

\title{Stabilizer States as a Basis for Density Matrices}
\author{Simon J.\ Gay\\
School of Computing Science,
University of Glasgow, UK}
\begin{document}
\maketitle

\begin{abstract}
  We show that the space of density matrices for $n$-qubit states,
  considered as a $(2^n)^2$-dimensional real vector space, has a basis
  consisting of density matrices of stabilizer states. We describe an
  application of this result to automated verification of quantum
  protocols.
\end{abstract}

\section{Definitions and Results}

We are working with the stabilizer formalism \cite{GottesmanD:claqec},
in which certain quantum states on sets of qubits are represented by
the intersection of their stabilizer groups with the group generated
by the Pauli operators. The stabilizer formalism is defined, explained
and illustrated in a substantial literature; good introductions are
given by Aaronson and Gottesman \cite{AaronsonS:impssc} and Nielsen
and Chuang \cite[Sec.~10.5]{NielsenMA:quacqi}.

In this paper we only need to use the following facts about stabilizer
states.
\begin{enumerate}
\item The standard basis states are stabilizer states.
\item The set of stabilizer states is closed under application of
  Hadamard ($\gH$), Pauli ($\gX$, $\gY$, $\gZ$), controlled not
  ($\gCnot$), and phase ($\gP = \matr{1}{0}{0}{i}$) gates.
\item The set of stabilizer states is closed under tensor product.
\end{enumerate}

\begin{notation}
  Write the standard basis for $n$-qubit states as $\{ \ket{x} \mid
  0\leqslant x < 2^n \}$, considering numbers to stand for their
  binary representations. We omit normalization factors when writing
  quantum states.
\end{notation}

\begin{definition}
Let $\ghz_n = \ket{0} + \ket{2^n-1}$ and
$\ighz_n = \ket{0} + i\ket{2^n-1}$, as $n$-qubit states.
\end{definition}

\begin{lemma}\label{lem:ghz}
For all $n$, $\ghz_n$ and $\ighz_n$ are stabilizer states.
\end{lemma}
\begin{proof}
  By induction on $n$. For the base case ($n=1$), we have that
  $\ket{0} + \ket{1}$ and $\ket{0} + i\ket{1}$ are stabilizer states,
  by applying $\gH$ and then $\gP$ to $\ket{0}$.

  For the inductive case, $\ghz_n$ and $\ighz_n$ are obtained from
  $\ghz_{n-1}\otimes\ket{0}$ and $\ighz_{n-1}\otimes\ket{0}$,
  respectively, by applying $\qgate{CNot}$ to the two rightmost
  qubits.\qed
\end{proof}

\begin{lemma}\label{lem:sumstabilizer}
  If $0 \leqslant x,y < 2^n$ and $x\neq y$ then
  $\ket{x} + \ket{y}$ and
  $\ket{x} + i\ket{y}$ are stabilizer states.
\end{lemma}
\begin{proof}
  By induction on $n$. For the base case ($n=1$), the closure
  properties imply that $\ket{0} + \ket{1}$, $\ket{0} + i\ket{1}$ and
  $\ket{1} + i\ket{0} = \ket{0} - i\ket{1}$ are stabilizer states.

  For the inductive case, consider the binary representations of $x$
  and $y$. If there is a bit position in which $x$ and $y$ have the
  same value $b$, then $\ket{x} + \ket{y}$ is the tensor product of
  $\ket{b}$ with an $(n-1)$-qubit state of the form $\ket{x'} +
  \ket{y'}$, where $x'\neq y'$. By the induction hypothesis,
  $\ket{x'}+\ket{y'}$ is a stabilizer state, and the conclusion
  follows from the closure properties. Similarly for $\ket{x} + i\ket{y}$.

  Otherwise, the binary representations of $x$ and $y$ are
  complementary bit patterns. In this case, $\ket{x}+\ket{y}$ can be
  obtained from $\ghz_n$ by applying $\gX$ to certain qubits. The
  conclusion follows from Lemma~\ref{lem:ghz} and the closure
  properties. The same argument applies to $\ket{x} + i\ket{y}$, using
  $\ighz_n$.\qed
\end{proof}

\begin{theorem}
  The space of density matrices for $n$-qubit states, considered as a
  $(2^n)^2$-dimensional real vector space, has a basis consisting of
  density matrices of $n$-qubit stabilizer states.
\end{theorem}
\begin{proof}
This is the space of Hermitian matrices and its obvious basis is the union of
\begin{equation}\label{eq1}
\{ \ket{x}\bra{x} \mid 0 \leqslant x < 2^n \}
\end{equation}
\begin{equation}\label{eq2}
\{ \ket{x}\bra{y} + \ket{y}\bra{x} \mid 0 \leqslant x < y < 2^n \}
\end{equation}
\begin{equation}\label{eq3}
\{ -i\ket{x}\bra{y} + i\ket{y}\bra{x} \mid 0 \leqslant x < y < 2^n \}.
\end{equation}
Now consider the union of
\begin{equation}\label{eq4}
\{ \ket{x}\bra{x} \mid 0 \leqslant x < 2^n \}
\end{equation}
\begin{equation}\label{eq5}
\{ (\ket{x}+\ket{y})(\bra{x}+\bra{y}) \mid 0 \leqslant x < y < 2^n \}
\end{equation}
\begin{equation}\label{eq6}
\{ (\ket{x}+i\ket{y})(\bra{x}-i\bra{y}) \mid 0 \leqslant x < y < 2^n \}.
\end{equation}
This is also a set of $(2^n)^2$ states, and it spans the space because
we can obtain states of forms (\ref{eq2}) and (\ref{eq3}) by subtracting
states of form (\ref{eq4}) from those of forms (\ref{eq5}) and
(\ref{eq6}). Therefore it is a basis, and by
Lemma~\ref{lem:sumstabilizer} it consists of stabilizer states.\qed
\end{proof}

\section{Applications}
The stabilizer formalism can be used to implement an efficient
classical simulation of quantum computation, if quantum operations are
restricted to those under which the set of stabilizer states is closed
--- i.e.\ the Clifford group operations. This result is the
Gottesman-Knill Theorem \cite{GottesmanD:stacqe}. Gay, Nagarajan and
Papanikolaou \cite{GayS:qmcmcq,GayS:spevq,PapanikolaouN:modcqp} have
applied it to formal verification of quantum systems, adapting
model-checking techniques \cite{ClarkeEM:modc} from classical computer
science. A simple example of model-checking is the following.

Consider a quantum teleportation protocol as a system with one qubit
as input and one (different) qubit as output; call this system
$\mathit{Teleport}$. Also consider the one-qubit identity operator
$I$. Then the specification that teleportation should satisfy is that
$\mathit{Teleport} = I$, where equality means the same transformation
of a one-qubit state. The aim of model-checking in this context is to
automatically verify that this specification is satisfied, by
simulating the operation of the two systems on all possible
inputs. For this to be possible, the teleportation protocol is first
expressed in a formal modelling language analogous to a programming
language.

The approach of Gay, Nagarajan and Papanikolaou reduces the problem to
that of simulating teleportation on stabilizer states as inputs, which
can be done efficiently because the teleportation protocol itself only
uses Clifford group operations. Correctness of teleportation on
stabilizer states is interpreted as evidence for, although not proof
of, correctness on arbitrary states. Now, however, we can draw a
stronger conclusion.

The teleportation protocol defines a superoperator; this can be
proved, for example, by the techniques of Selinger
\cite{SelingerP:towqpl}, who uses superoperators to define the
semantics of a programming language that can certainly express
teleportation. Superoperators, among other properties, are linear
operators on the space of density matrices. To check equivalence of
two superoperators, it is therefore sufficient to check that they have
the same effect on all elements of a particular basis. By taking a
basis that consists of stabilizer states, this can be done
efficiently.

Moreover, the number of stabilizer states on $n$ qubits is
approximately $2^{(n^2)/2}$, whereas the dimensionality of the space of
$n$-qubit density matrices is only $(2^n)^2 = 2^{2n}$. It is therefore
more efficient to model-check on a basis than on all stabilizer
states.
   
\section{Conclusion}
We have proved that the space of $n$-qubit density matrices has a
basis consisting of stabilizer states, and explained how this result
can be used to improve the efficiency and strengthen the results of
model-checking for quantum systems.

\bibliographystyle{simonplain}
\bibliography{main}

\end{document}